\documentclass[preprint,12pt]{elsarticle}

\usepackage{amssymb,amsmath,amsthm,graphicx}
\usepackage{enumerate}
\usepackage{float}
\usepackage{stmaryrd,mathbbol,mathrsfs} 
\usepackage{multicol}
\usepackage{float}
\usepackage{stackrel}
\usepackage{color}
\usepackage{tikz}
\usetikzlibrary{matrix,arrows}
\usepackage{capt-of}
\usepackage{array}

\newtheorem{theorem}{Theorem}[section]
\newtheorem*{theorem*}{Theorem}
\newtheorem{lemma}[theorem]{Lemma}
\newtheorem{proposition}[theorem]{Proposition}
\newtheorem{corollary}[theorem]{Corollary}

\theoremstyle{remark}

\theoremstyle{definition}
\newtheorem{example}[theorem]{Example}

\journal{}

\makeatletter
\def\ps@pprintTitle{%
 \let\@oddhead\@empty
 \let\@evenhead\@empty
 \def\@oddfoot{}%
 \let\@evenfoot\@oddfoot}
\makeatother

\begin{document}

\begin{frontmatter}



\title{Description of sup- and inf-preserving aggregation functions via families of clusters in data tables\footnote{Preprint of an article published by Elsevier in the Information Sciences 400-401 (2017), 173-183. It is available online at: \newline www.sciencedirect.com/science/article/pii/S0020025517305510}}


\author[up]{Radom\'ir Hala\v{s}}\ead{radomir.halas@upol.cz}
\author[up,stu]{Radko Mesiar}\ead{radko.mesiar@stuba.sk}
\author[up,sav]{Jozef P\'ocs}\ead{pocs@saske.sk}

\address[up]{Palack\'y University Olomouc, Faculty of Science, Department of Algebra and Geometry, 17. listopadu 12, 771 46 Olomouc, Czech Republic}

\address[stu]{Department of Mathematics and Descriptive Geometry, Faculty of Civil Engineering, Slovak University of Technology in Bratislava, Radlinsk\'eho 11, 810 05 Bratislava 1, Slovakia}

\address[sav]{Mathematical Institute, Slovak Academy of Sciences,\\ Gre\v s\'akova 6, 040 01 Ko\v sice, Slovakia}

\begin{abstract}
Connection between the theory of aggregation functions and formal concept analysis is discussed and studied, thus filling 
a gap in the literature by building a bridge between these two theories, one of them living in the world of data
fusion, the second one in the area of data mining. We show how Galois connections can be used to describe an important class of aggregation
functions preserving suprema, and, by duality, to describe aggregation functions preserving infima. Our discovered method gives an elegant and
complete description of these classes. Also possible applications of our results within certain biclustering fuzzy FCA-based methods are discussed.
\end{abstract}


\end{frontmatter}

\section{Introduction}

Among several theoretical tools applied in data mining, an important role is
played by aggregation functions. Recall, for example, the application of
Choquet and Sugeno integrals and other aggregation functions discussed in
\cite{STKPL}. 
Note that aggregation functions were originally introduced to act on real
intervals, for details we refer the reader to the comprehensive monographs \cite{BPC,Grabisch et al 2009} and \cite{KMP}. 

However, recently the
aggregation on posets, and in particular on lattices, has became a rapidly
growing topic, especially due to applications in information sciences,
see e.g. \cite{Couceiro}, \cite{Grabisch et al 2009} etc.
This trend was
stressed at the international conference ABLAT (Aggregation on Bounded
LATtices) held in 2014 in Trabzon, Turkey, among others. Aggregation on
lattices exploiting several algebraic results is often bringing new lights
also into the standard real-valued aggregation techniques. For example, though
the Sugeno integral was introduced in 1974 \cite{Sug74}, and discussed in
many papers and monographs, only recently, based on algebraic look, it was
shown that it is, in fact, an aggregation function which preserves congruences
\cite{HMP1, HMP3}. This allows to apply the Sugeno integral consistently also when
we change a numerical scale into linguistic scale, for example. Also quite recently, 
the essential progress with respect to understanding how aggregation functions can be generated 
has been achieved, see \cite{HMP2,HP1,HP2}. For more details
on recent applications of aggregation functions we refer the reader to \cite{T1,T2,Kara2,T4,Kara1,Lei,MSY,MY,MZ1,Ze1,MZ2}.

Another useful method used in data mining during last decades is Formal Concept Analysis (FCA, in short) \cite{GW}. It is a theory where data are analysed by means of conceptual structures among data sets. Mathematically,  FCA is based on the notion of a formal context which is represented by two sets, objects and attributes, and by a binary relation between the set of objects and the set of attributes representing the relationship between them. As a result of the process, we obtain so-called formal concepts
which correspond to the maximal rectangles in the data set. The set of formal concepts has a structure of a complete lattice (called concept lattice) consisting of all conceptual abstractions (concepts) combining subsets of objects with subsets of shared  attributes. FCA has been proved to be an effective tool in many areas of science, besides decision making it has been extensively applied to fields such as knowledge discovery, information retrieval, software engineering etc.

Classical FCA method is used for a binary case where we can sharply decide whether or not a given object has a given attribute. On the other hand, there are many natural practical examples where this most simple case is not appropriate. Namely, in many concrete situations the relation between objects and attributes is not crisp, and then many-valued or fuzzy description of object-attribute model is more convenient. Therefore, when dealing with imprecise data, uncertainty or even when the information is not complete, this more general setting of FCA method has become an important research topic in the recent years. For more details we recommend the reader the papers \cite{AL,AK,B01,BuPPINS,MedN3,HP3,Kardos,Kr1,MedN2,MedN1,MO,Med09,Pocs,P3}. 

Concept lattices can be viewed from another important equivalent way. Namely, they correspond to Galois-closed sets with respect 
to a Galois connection, both in monotone as well as in antitone setting, induced by the incidence relation between the sets of objects and attributes. Remark that Galois connections play a fundamental role in mathematics because of their universality. 
To be more specific, Galois connections represent a structure-preserving passage between two worlds, the one living on the object side, the second on the attribute side. Consequently, they are inherent with respect to human thinking in a sense that they allow to closely connect certain quite different worlds of hierarchical structures. Order-theoretically, Galois connections consist of two order-preserving maps whose composition yields two operators on the respective structures, one closure operator, and the second being a kernel operator. In other words, the two hierarchies living in two different worlds can be transported to each other. Such an adjoint situation has an advantage that the knowledge about one of the worlds can be used to gain the information in the second one. Remind that historically, the classical Galois theory has been used for solving the problem of solvability of algebraic equations. Besides this one can find many other applications of Galois-correspondences in almost all branches of mathematics and applied science. As an example, recall the link between conjunctions and
implications in fuzzy logics \cite{Haj}. For more details on Galois connections we recommend the survey paper by Ern\'e et al. \cite{Erne}.

By authors' knowledge, so far no essential connection between the theory of aggregation functions and FCA has been developed. The purpose of this paper is to fill this gap by building a bridge between these two theories, one of them living in the world of data fusion, the second one in the area of data mining. 
Our aim is to show how Galois connections can be used to describe
important classes of aggregation functions. Note that the majority of
aggregation functions exploited in applications are considered to preserve
suprema (e.g., for maximization problems on real intervals) or infima (e.g.,
for minimization problems). Moreover, on real intervals, simultaneous preservation
of suprema and infima by an aggregation function means its continuity. In our
paper, we focus on the case of sup-preserving aggregation functions only,
while the related results for inf-preserving aggregation functions can be
obtained by duality. As shown below, our method
gives an elegant and complete description of this class.

\section{Preliminaries}

We assume that the reader is familiar with the basic notions and terminology of partially ordered sets, especially lattice theory, cf. \cite{Blyth} or \cite{Gratzer}.

The direct product of an indexed system $\{L_i\mid i\in I\}$ of lattices is defined in the usual way; we apply the notation $\prod_{i\in I} L_i$. If $L_i=L$ for all $i\in I$, the symbol $L^I$ denotes the direct power of a lattice $L$ or we use $L^n$ provided $I=\{1,\dots,n\}$. The elements of the direct product will be denoted by the bold symbols and for $\mathbf{x}\in \prod_{i\in I} L_i$, $\mathbf{x}(i)$ is the $i$-th component of $\mathbf{x}$ in the lattice $L_i$.
Recall that the direct product of lattices forms a complete lattice if and only if all of them are complete lattices.

Given a partially ordered set $P$, a \textit{closure operator} on $P$ is a self-map $c\colon P\to P$ which is monotone, extensive and idempotent. More precisely, this means the following conditions for all $x,y\in P$:
\begin{enumerate}
\item $c(x)\leq c(y)$ provided $x\leq y$ 
\item $x\leq c(x)$
\item $c(x)=c(c(x))$.
\end{enumerate}
The notion of an interior operator is defined dually, i.e., an \textit{interior operator} on $P$ is a mapping $i\colon P\to P$ which is monotone, intensive (i.e., $i(x)\leq x$ for all $x\in P$) and idempotent. 

In the case of complete lattices, the notions of closure and interior operators are closely related to that of closure systems and interior systems, respectively. Given a complete lattice $L$, a \textit{closure system} on $L$ is a subset $S\subseteq L$ closed under arbitrary infima, i.e., 
$$ X\subseteq S\quad \Longrightarrow\quad \bigwedge X \in S.$$ 
Dually, an \textit{interior system} on $L$ is a subset $T\subseteq L$ closed under arbitrary suprema, i.e.,
$$ X\subseteq T\quad \Longrightarrow\quad \bigvee X \in T.$$ 
Let us remark that as $\bigwedge \emptyset = 1$ and $\bigvee \emptyset =0$ holds in any complete lattice, the top element $1$ belongs to every closure system and analogously, every interior system contains the bottom element $0$.

It is the well-known fact that every closure operator gives rise to a closure system and vice versa. 
For a closure system $S$ on $L$, one can define $c_S\colon L\to L$ by $c_S(x)=\bigwedge\{y\in S\mid x\leq y\}$. Conversely, given a closure operator $c$ on $L$, the set $\mathrm{Fix}(c)=\{x\in L\mid c(x)=x\}$ of fixed points forms a closure system on $L$. In this case $\mathrm{Fix}(c_S)=S$ and $c_{\mathrm{Fix}(c)}=c$.  

A similar correspondence holds between interior operators and interior systems, in this case $i_T(x)=\bigvee \{y\in T\mid y\leq x\}$ represents the interior operator corresponding to an interior system~$T$.

Further, we recall the definition and basic properties of monotone Galois connections. The results presented in this section can be found in several sources, however in a non-compact form, c.f. \cite{Blyth} or \cite{Erne}. In order to make the paper as self-contained as possible, we provide the necessary results in a modified comprehensive form together with their proofs.

Let $P$, $Q$ be two posets. We say that a pair of mappings $(f,g)$,  $f\colon P\to Q$ and $g\colon Q\to P$ forms a \textit{monotone Galois connection} if for all $x\in P$ and $y\in Q$ it holds
\begin{equation}\label{eq_res}
f(x)\leq y \quad\quad\mbox{iff}\quad\quad x\leq g(y).
\end{equation} 
Then $f$ is called the \textit{lower adjoint} of $g$, while $g$ is referred to as the \textit{upper adjoint} of $f$. Note that given a mapping $f$, there is at most one upper adjoint $g$ satisfying \eqref{eq_res}. To see this, consider two such mappings $g_1$ and $g_2$. From \eqref{eq_res} we immediately infer $x\leq g_1(y)$ if and only if $x\leq g_2(y)$, i.e., the sets of lower bounds of $g_1(y)$ and $g_2(y)$ coincide, implying that $g_1(y)=g_2(y)$. The uniqueness of the lower adjoint (if it exists) corresponding to a given mapping $g$ can be shown similarly.

Let $(f,g)$ be a monotone Galois connection between $P$ and $Q$. Then $f$ and $g$ are both monotone and they satisfy $ x\leq g(f(x))$ and $f(g(y))\leq y$ for all $x\in P$ and $y\in Q$. Indeed, from \eqref{eq_res} we easily obtain that $f(x)\leq f(x)$ implies $x\leq g(f(x))$. Consequently, applying \eqref{eq_res} again, the inequality $x_1\leq x_2\leq g(f(x_2))$ yields $f(x_1)\leq f(x_2)$, i.e., $f$ is monotone. 

Consequently, we obtain the following important property: 
\begin{equation}\label{e2}
f\circ g\circ f=f \quad\quad\mbox{and}\quad\quad g\circ f\circ g=g.
\end{equation} 
To see this, using \eqref{eq_res}, from $g(f(x))\leq g(f(x))$ we obtain $f(g(f(x)))\leq f(x)$. On the other hand, $x\leq g(f(x))$ and the monotonocity of $f$ yield $f(x)\leq f(g(f(x)))$.

Let $L$ be a complete lattice. For an element $a\in L$, $(a\rangle$ denotes the principal ideal generated by $a$, i.e., $(a\rangle=\{x\in L\mid x\leq a\}$. Dually, $\langle a)=\{x\in L\mid a\leq x\}$ denotes the principal filter generated by $a$. A subset $L_1\subseteq L$ is called hereditary, or a down-set, if for every $x \in L $ and $x_1\in L_1$ with $x\leq x_1$ we have $x\in L_1$. The concept of an up-set is defined dually.

\begin{proposition}\label{prop1}
Let $L$ and $M$ be complete lattices and $f\colon L\to M$ be a mapping. The following conditions are equivalent.
\begin{enumerate}
\item $f$ is $\bigvee$-preserving.
\item The inverse image $f^{-1}\big((a\rangle\big)\subseteq L$ is a down-set for every $a\in M$.
\item There exists an upper adjoint mapping $g\colon M\to L$ of $f$.  
\end{enumerate}
\end{proposition}

\begin{proof}
$(1)\Rightarrow (2):$ Let $a\in M$ and $x\in f^{-1}\big((a\rangle\big)=U$. As $f$ is $\bigvee$-preserving, it is monotone. Thus, for $x_1\leq x$ we obtain $f(x_1)\leq f(x)\in (a\rangle$, which yields that $x_1\in U$, i.e., the set $U$ is hereditary. Further, put $b=\bigvee_{x\in U} x$. As $f(x)\leq a$ for all $x\in U$, we obtain 
$$ f(b)=f(\bigvee_{x\in U} x)=\bigvee_{x\in U}f(x)\leq a,$$
showing that $b\in U$. Since $U$ is a down-set and $b$ is its greatest element, we have $U=(b\rangle$.

$(2)\Rightarrow (3):$ For $a\in M$ put $g(a)=b$ where $b\in L$ is such that $(b\rangle=f^{-1}\big((a\rangle\big)$. Then obviously 
$$f(x)\leq y \quad \mbox{iff} \quad x\in f^{-1}\big((y\rangle\big) \quad\mbox{iff}\quad x\leq g(y).$$   

$(3)\Rightarrow (1):$ Assume that $f$ and $g$ fulfill \eqref{eq_res} and let $\{x_i\mid i\in I\}\subseteq L$ be a family of elements. We have already shown that $f$ is monotone. Then $x_i\leq \bigvee_{i\in I} x_i$ for all $i\in I$ and the monotonicity of $f$ yields  
$$\bigvee_{i\in I} f(x_i)\leq f\big(\bigvee_{i\in I} x_i\big).$$
Conversely, denoting $y=\bigvee_{i\in I} f(x_i)$, we have $f(x_i)\leq y$ for each $i\in I$. 
However, by \eqref{eq_res} this holds if and only if
$$ x_i\leq g(y), \forall i\in I \quad\mbox{iff}\quad \bigvee_{i\in I}x_i \leq g(y)\quad \mbox{iff}\quad f\big(\bigvee_{i\in I}x_i\big)\leq y,$$
which shows that $f$ is $\bigvee$-preserving.
\end{proof}

Using the dual arguments, one can similarly prove the following. 

\begin{proposition}\label{prop2}
Let $L$ and $M$ be complete lattices and $g\colon M\to L$ be a mapping. Then the following conditions are equivalent:
\begin{enumerate}
\item $g$ is $\bigwedge$-preserving.
\item The inverse image $g^{-1}\big(\langle b)\big)\subseteq M$ is an up-set for every $b\in L$.
\item There exists the lower adjoint mapping $f\colon L\to M$ of $g$.
\end{enumerate}
\end{proposition}

Hence, as a conclusion of the above propositions we obtain that given a $\bigvee$-preserving mapping $f\colon L\to M$, there is the unique $\bigwedge$-preserving mapping $g\colon M\to L$ such that $f$ and $g$ form a monotone Galois connection. 

For a mapping $f\colon L\to M$, let $\mathrm{Rng}(f)=\{f(x)\colon x\in L\}$ denotes its range.
The following important assertion provides an inner description of monotone Galois connections. 

\begin{proposition}\label{prop3}
Let $L$ and $M$ be complete lattices and $f\colon L\to M$, $g\colon M\to L$ be two mappings between them. If the pair $(f,g)$ forms a monotone Galois connection then 
\begin{enumerate}
\item the range $\mathrm{Rng}(f)$ is an interior system on $M$, 
\item the range $\mathrm{Rng}(g)$ is a closure system on $L$, 
\item $\mathrm{Rng}(f)$ and $\mathrm{Rng}(g)$ are isomorphic posets. 
\end{enumerate}

Conversely, let $S$ be a closure system on $L$ with the corresponding closure operator $c_S$, $T$ an interior system on $M$ with the corresponding interior operator $i_T$ and $\varphi\colon S\to T$ an isomorphism. Then the mappings $c_S\circ \varphi\colon L\to M$, $i_T\circ\varphi^{-1}\colon M\to L$ form a monotone Galois connection.
\end{proposition}

\begin{proof}
Obviously, $\mathrm{Rng}(f)$ is an interior system on $M$ since $f$ is $\bigvee$-preserving. Similarly, $\mathrm{Rng}(g)$ is a closure system on $L$. Further, as $f$ and $g$ satisfy \eqref{e2}, we have $f(g(y))=y$ for all $y\in T$ as well as $g(f(x))=x$ for all $x\in S$. Thus, we obtain that $f$ restricted to $S$ and $g$ restricted to $T$ are mutually inverse. As both these mappings are monotone, $S$ and $T$ are isomorphic posets.

Conversely, assume that $S$ and $T$ are isomorphic. Since $\varphi^{-1}(i_T(y))\in S$ and $\varphi(c_S(x))\in T$, from the basic properties of closure and interior operators we obtain 
$$ x\leq \varphi^{-1}(i_T(y))\quad \mbox{iff} \quad c_S(x)\leq c_S(\varphi^{-1}(i_T(y)))=\varphi^{-1}(i_T(y)),$$ which is equivalent to 
$$\varphi(c_S(x))=i_T(\varphi(c_S(x)))\leq i_T(y) \quad \mbox{iff} \quad \varphi(c_S(x))\leq y.$$

\end{proof}

In the sequel, if a pair of isomorphic closure-interior systems is considered, it is implicitly assumed that some isomorphism between them is also present. Moreover, we will consider ($n$-ary) aggregation functions
$f\colon L^{n}\to L$, i.e., functions which are characterized by the monotonicity and
boundary conditions. Hence $f$ is an aggregation function whenever $f(x_1,\dots,x_n)\leq
f(y_1,\dots,y_n)$ if $x_1\leq y_1,\dots,x_n\leq y_n$, and $f(0,\dots,0)=0, f(1,\dots,1)=1$. In
particular, each homomorphism $f\colon L^{n}\to L$ is an aggregation function.

\section{Sup-preserving and inf-preserving aggregation functions}

Recall that any aggregation function $f$ on a complete lattice $L$ fulfills the boundary conditions 
$$ f(0,\dots,0)=0 \quad\mbox{and}\quad f(1,\dots,1)=1,$$ 
where $0$ and $1$ denote the bottom and the top element of the lattice $L$. In order to apply Proposition \ref{prop3}, these conditions give the following basic characterization: 

\begin{lemma}\label{lem1}
Let $L$ be a complete lattice and $f\colon L^n\to L$ be an aggregation function.
Then $f$ is $\bigvee$-preserving if and only if $f(\mathbf{x})=\varphi(c_S(\mathbf{x}))$ for all $\mathbf{x}\in L^n$, where $\varphi\colon S\to T$ is an isomorphism between a closure system $S\subseteq L^n$ and an interior system $T\subseteq L$ such that $1\in T$, and $c_S\colon L^n\to L^n$ is the closure operator corresponding to $S$.

Similarly, $f$ is $\bigwedge$-preserving if and only if $f(\mathbf{x})=\varphi(i_T(\mathbf{x}))$ for all $\mathbf{x}\in L^n$, where $\varphi\colon T\to S$ is an isomorphism between an interior system $T\subseteq L^n$ and a closure system $S\subseteq L$ such that $0\in S$, and $i_T\colon L^n \to L^n$ is the interior operator corresponding to $T$. 
\end{lemma}

\begin{proof}
Assume that $f\colon L^n\to L$ is a $\bigvee$-preserving aggregation function. According to Proposition \ref{prop3}, $f$ is determined by isomorphic closure system $S\subseteq L^n$ and interior system $T\subseteq L$ via some isomorphism $\varphi\colon S\to T$. As $f(1,\dots,1)=1$ and $(1,\dots,1)\in S$ (every closure system contains the top element), we obtain $$f(1,\dots,1)=\varphi(c_S(1,\dots,1))=\varphi(1,\dots,1)=1$$
and it follows that $1\in T$. 

Conversely, assume that $f(\mathbf{x})=\varphi(c_S(\mathbf{x}))$. Obviously, $f$ is $\bigvee$-preserving. Since the closure $c_S(0,\dots,0)$ equals to $\bigwedge S$, the least element of $S$, we obtain
$$ f(0,\dots,0)=\varphi(c_S(0,\dots,0))=\varphi(\bigwedge S)=0.$$ Note that the last equality follows from the fact that $\varphi$ maps the least element of $S$ into the least element of $T$, which is $0$. Dually, for the greatest element of $L^n$ we obtain 
$$ f(1,\dots,1)=\varphi(c_S(1,\dots,1))=\varphi(1,\dots,1)=1,$$ since $\varphi$ maps the greatest element of $S$ into the greatest element $1\in T$. 

The assertion concerning the $\bigwedge$-preserving aggregation functions can be proved analogously. 
\end{proof}

Although the previous lemma provides the basic inner characterization of the $\bigvee$-preserving as well as the $\bigwedge$-preserving mappings, closure and interior systems on a direct power $L^n$ need not be so transparent. In what follows we try to find a similar characterization, however with respect to possibly more simple factors.

Obviously, any $\bigvee$-preserving aggregation function $f\colon L^n \to L$ is decomposable in the following way 
\begin{equation}\label{eq2a}
 f(\mathbf{x})=\bigvee_{i=1}^n f_i(\mathbf{x}(i)), \quad\mbox{for all}\ \mathbf{x}\in L^n,
\end{equation} 
where $f_i\colon L\to L$ for all $i\in \{1,\dots,n\}$ is $\bigvee$-preserving. In this case, for each $i\in \{1,\dots,n\}$ the function $f_i$ is the lower adjoint of $g\circ \pi_i\colon L \to L$, where $g$ is the upper adjoint of $f$ and $\pi_i\colon L^n\to L$ denotes the $i$-th projection.

Conversely, given any system of $\{f_i\colon L\to L\mid 1\leq i\leq n\}$ of $\bigvee$-preserving mappings, the formula \eqref{eq2a} determines a $\bigvee$-preserving function $f$. In addition, $f$ being an aggregation function, the functions of the system have to satisfy $\bigvee_{i=1}^n f_{i}(1)=1$. 
Hence, there is a one-to-one correspondence between the family of all $n$-ary $\bigvee$-preserving aggregation functions on a complete lattice $L$ and the family whose elements consist of $n$ pairs of closure-interior system pairs $\{(S_i,T_i)\mid 1\leq i\leq n\}$ on $L$, satisfying $\bigvee_{i=1}^n \top_i=1$, where $\top_i$ is a greatest element in $T_i$.

Notice that a similar characterization can be also applied for $\bigwedge$-preserving aggregation functions.

\begin{example}\label{ex1}
Consider the six-element lattice $L=\{0,1,a,b,c,d\}$ whose Hasse diagram is depicted in Fig. \ref{fig1}. 

\begin{figure}
\begin{center}
\includegraphics[scale=1]{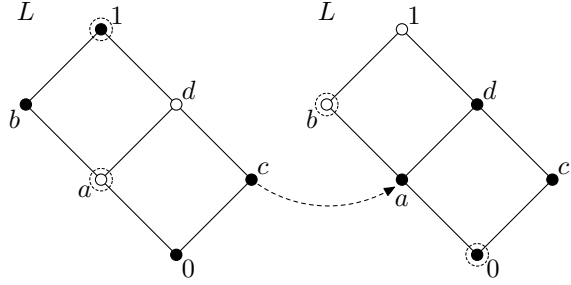}
\caption{An example of isomorphic closure-interior systems.}
\label{fig1}
\end{center}
\end{figure}

In order to generate a binary $\bigvee$-preserving aggregation function, consider two pairs of isomorphic closure-interior systems $(S_1,T_1)$ and $(S_2,T_2)$ on $L$. In this case $S_1=\{1,b,c,0\}$ and $T_1=\{d,a,c,0\}$. Both systems are indicated by filled circles, where $S_1$ is depicted on the left side and $T_1$ on the right side of Fig.~\ref{fig1}. The systems $S_2$ and $T_2$ are denoted by slashed circles, particularly $S_2=\{1,a\}$ and $T_2=\{b,0\}$. Let us note that for the pair $(S_2,T_2)$, the corresponding isomorphism $\varphi_2$ between these two systems is unique. However, considering the pair $(S_1,T_1)$, the corresponding isomorphism $\varphi_1\colon S_1\to T_1$ is indicated by the arrow, i.e., $\varphi_1(c)=a$. Then we necessarily have $\varphi_1(b)=c$, and the top and the bottom elements are mapped into their respective counterparts.
As the top elements of $T_1$ and $T_2$ satisfy $\top_1\vee\top_2=d\vee b=1 $, the considered two pairs of isomorphic closure-interior systems induce a $\bigvee$-preserving aggregation function $f$, the values of which are given in Table \ref{tab1}.
\begin{table}
\begin{center}
\begin{tabular}{|c|c c c c c c|}
\hline
$x \setminus y$ & $0$ & $a$ & $b$ & $c$ & $d$ & $1$ \\
\hline
$0$             & $0$ & $0$ & $b$ & $b$ & $b$ & $b$ \\
$a$             & $c$ & $c$ & $1$ & $1$ & $1$ & $1$ \\
$b$             & $c$ & $c$ & $1$ & $1$ & $1$ & $1$ \\
$c$             & $a$ & $a$ & $b$ & $b$ & $b$ & $b$ \\
$d$             & $d$ & $d$ & $1$ & $1$ & $1$ & $1$ \\
$1$             & $d$ & $d$ & $1$ & $1$ & $1$ & $1$ \\
\hline
\end{tabular}
\caption{The table corresponding to the values $f(x,y)$.}
\label{tab1}
\end{center}
\end{table}
For example, the value $f(c,d)$ is calculated in the following way. First, the closures $c_{S_1}(c)$ and $c_{S_2}(d)$ are determined. In the first case, $c_{S_1}(c)$ is the smallest element in $S_1$ which is above $c$ and, obviously, it is $c$ itself. In the second case, $c_{S_2}(d)$ is the smallest element in $S_2$ which is above $d$, which is $1$. Finally, $$f(c,d)=\varphi_1(c_{S_1}(c))\vee \varphi_2(c_{S_2}(d))=\varphi_1(c)\vee \varphi_2(1)=a\vee b=b.$$  

\end{example}

Further, we describe a decomposition in the case when an underlying complete lattice is a direct product of an indexed system of complete lattices.

\begin{theorem}\label{thm11}
Let $\{L_{\lambda}\mid \lambda\in \Lambda\}$ and $\{M_{\gamma}\mid \gamma\in \Gamma\}$ be indexed families of complete lattices. Then there is a $\bigvee$-preserving mapping $f\colon\prod_{\lambda\in \Lambda}L_{\lambda}\to \prod_{\gamma\in\Gamma}M_{\gamma}$ if and only if there is a system $\{f_{\lambda\gamma}\colon L_{\lambda}\to M_{\gamma}\mid \lambda\in\Lambda, \gamma\in\Gamma\}$ of $\bigvee$-preserving mappings such that 
\begin{equation}\label{eq1}
f\big(\mathbf{x}\big)(\gamma)=\bigvee_{\lambda\in\Lambda}f_{\lambda\gamma}(\mathbf{x}(\lambda)),\quad \mbox{for all }\mathbf{x}\in\prod_{\lambda\in\Lambda}L_{\lambda}. 
\end{equation}
\end{theorem}

\begin{proof}
First, assume that there is the above system $\{f_{\lambda\gamma}\mid \lambda\in\Lambda, \gamma\in\Gamma\}$ of $\bigvee$-preserving mappings. Let $\{\mathbf{x}_i\mid i\in I\}\subseteq \prod_{\lambda\in \Lambda}L_{\lambda} $ be a family of elements and let $f$ be defined by \eqref{eq1}.
Then for each $\gamma\in \Gamma$ 
$$ f\big(\bigvee_{i\in I}\mathbf{x}_i\big)(\gamma) = \bigvee_{\lambda\in\Lambda}f_{\lambda\gamma}\big(\bigvee_{i\in I}\mathbf{x}_i(\gamma)\big)=\bigvee_{\lambda\in\Lambda}\bigvee_{i\in I}f_{\lambda\gamma}\big(\mathbf{x}_i(\gamma)\big),$$ 
where the last equality follows from the fact that for each $\lambda\in\Lambda$ and $\gamma\in \Gamma$ the mapping $f_{\lambda\gamma}$ is $\bigvee$-preserving. However, due to basic properties of the supremum operation, we obtain 
$$ \bigvee_{\lambda\in\Lambda}\bigvee_{i\in I}f_{\lambda\gamma}\big(\mathbf{x}_i(\gamma)\big)=\bigvee_{i\in I}\bigvee_{\lambda\in\Lambda}f_{\lambda\gamma}\big(\mathbf{x}_i(\gamma)\big)=\bigvee_{i\in I} f\big(\mathbf{x}_i\big)(\gamma) = \big(\bigvee_{i\in I}f(\mathbf{x}_i)\big)(\gamma).$$ 

Conversely, assume that $f\colon\prod_{\lambda\in \Lambda}L_{\lambda}\to \prod_{\gamma\in\Gamma}M_{\gamma}$ is a $\bigvee$-preserving mapping. 
For $\lambda\in \Lambda$ and $\gamma\in\Gamma$ define $f_{\lambda\gamma}\colon L_{\lambda}\to M_{\gamma}$ as $f_{\lambda\gamma}(x)=f(\mathbf{0}_{\lambda,x})(\gamma)$, where $\mathbf{0}_{\lambda,x}\in \prod_{\lambda\in\Lambda}L_{\lambda}$ is given by $\mathbf{0}_{\lambda,x}(\xi)=x$ if $\xi=\lambda$ and $\mathbf{0}_{\lambda,x}(\xi)=0$ otherwise. It is easily seen that every $f_{\lambda\gamma}$ is $\bigvee$-preserving.
Further,  for all $\gamma\in\Gamma$ we obtain
$$ f\big(\mathbf{x}\big)(\gamma)=f\big(\bigvee_{\lambda\in\Lambda}\mathbf{0}_{\lambda,\mathbf{x}(\lambda)}\big)(\gamma)=\bigvee_{\lambda\in\Lambda} f\big(\mathbf{0}_{\lambda,\mathbf{x}(\lambda)}\big)(\gamma)=\bigvee_{\lambda\in\Lambda}f_{\lambda\gamma}(\mathbf{x}(\lambda)).$$

\end{proof}

\begin{corollary}
Let $f\colon\prod_{\lambda\in \Lambda}L_{\lambda}\to \prod_{\gamma\in\Gamma}M_{\gamma}$ be a $\bigvee$-preserving mapping defined by \eqref{eq1}. Then its $\bigwedge$-preserving adjoint $g\colon\prod_{\gamma\in\Gamma}M_{\gamma}\to \prod_{\lambda\in\Lambda}L_{\lambda}$ is given by 
\begin{equation}\label{eq2}
g\big(\mathbf{x}\big)(\lambda)=\bigwedge_{\gamma\in\Gamma}g_{\lambda\gamma}(\mathbf{y}(\gamma)),\quad \mbox{for all }\mathbf{y}\in\prod_{\gamma\in\Gamma}M_{\lambda}, 
\end{equation}
where for each $\lambda\in\Lambda$ and $\gamma\in\Gamma$, $g_{\lambda\gamma}$ is the upper adjoint of $f_{\lambda\gamma}$.  
\end{corollary}

\begin{proof}
Let $f$, $g$ be the mappings defined by \eqref{eq1} and \eqref{eq2}, respectively.
Then for any $\mathbf{x}\in\prod_{\lambda\in\Lambda}L_{\lambda}$ and $\mathbf{y}\in\prod_{\gamma\in\Gamma}M_{\gamma}$ we obtain 
$$ f(\mathbf{x})\leq \mathbf{y} \quad \mbox{iff} \quad f\big(\mathbf{x}\big)(\gamma)=\bigvee_{\lambda\in\Lambda}f_{\lambda\gamma}(\mathbf{x}(\lambda))\leq \mathbf{y}(\gamma),\; \forall \gamma\in\Gamma,$$
which is equivalent to 
$f_{\lambda\gamma}(\mathbf{x}(\lambda))\leq \mathbf{y}(\gamma)$, for all $\lambda\in\Lambda$ and for all $\gamma\in\Gamma$. However, by \eqref{eq_res} this is equivalent to $\mathbf{x}(\lambda)\leq g_{\lambda\gamma}(\mathbf{y}(\gamma))$ for all $\lambda\in\Lambda$ and $\gamma\in \Gamma$, which holds if and only if 
$$ \mathbf{x}(\lambda)\leq \bigwedge_{\gamma\in\Gamma}g_{\lambda\gamma}(\mathbf{y}(\gamma))=g\big(\mathbf{y}\big)(\lambda),\; \forall \lambda\in \Lambda\quad \mbox{iff}\quad \mathbf{x}\leq g(\mathbf{y}).$$ This shows that $f$ and $g$ satisfy \eqref{eq_res}.
\end{proof}

In particular cases, the last two results can be significantly strengthened. For
example, consider an associative symmetric aggregation function
$f\colon ([0,1]^{n})^{2}\to [0,1]^{n}$ with a neutral element $e=(1,\dots,1)$, i.e., $f$ is a
triangular norm \cite{KMP} of the product lattice $[0,1]^n$. Then, due to \cite{BM},
$f$ is $\bigvee$-preserving ($\bigwedge$-preserving) if and only if $f$ is a product of classical
$\bigvee$-preserving ($\bigwedge$-preserving) t-norms, i.e., there are triangular norms
$f_i\colon [0,1]^{2}\to [0,1]$, $i=1,\dots,n$, which are $\bigvee$-preserving ($\bigwedge$-preserving), and
$f(\mathbf{x},\mathbf{y})=(f_1(x_1,y_1),\dots,f_n(x_n,y_n))$.

\begin{example}
Consider the lattice $L$ from Example \ref{ex1}. It can be easily seen that $L\cong \mathbf{3}\times \mathbf{2}$, where $\mathbf{3}=\{0,1,2\}$ and $\mathbf{2}=\{0,1\}$ denote the three-element and two-element chains respectively, with the usual order. We present the decomposition of the $\bigvee$-preserving mapping $f_1$ determined by the closure-interior system $(S_1,T_1)$ from Example \ref{ex1}. According to Theorem \ref{thm11}, there is a system of $\bigvee$-preserving mappings between the particular direct factors. The resulting decomposition via closure-interior systems is given in Fig. \ref{fig2}.

\begin{figure}
\begin{center}
\includegraphics[scale=1]{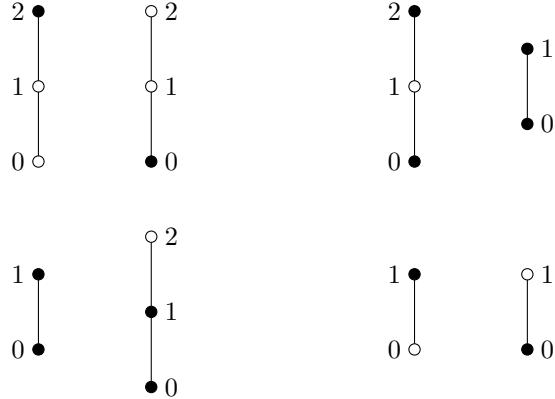}
\caption{Decomposition of the mapping corresponding to $(S_1,T_1)$ pair of Example \ref{ex1}.}
\label{fig2}
\end{center}
\end{figure}

\end{example}

In the finite case, every lattice is a direct product of directly indecomposable lattices. However, the structure of directly indecomposable lattices may be difficult, e.g., it is evident that every lattice with a prime number of elements is directly indecomposable. 
Hence, in order to obtain a representation of $\bigvee$-preserving aggregation functions with respect to possibly simpler lattices, we turn our attention to the so-called subdirect products, cf. \cite{Gratzer}.

Recall that a lattice $L$ is a subdirect product of an indexed family $\{L_i\mid i\in I\}$ of lattices if
\begin{enumerate}
\item $L$ is a sublattice of $\prod_{i\in I}L_i$, 
\item $\pi_i(L)=L_i$ for each $i\in I$, i.e., each coordinate projection $\pi_i$ maps $L$ onto the corresponding factor.
\end{enumerate}
In such a case we shall write $L\leq \prod_{i\in I}L_i$.  
A subdirect representation of a lattice is any embedding whose image is a subdirect product.
A lattice $L$ is said to be subdirectly irreducible if $\left|L\right|>1$ and all subdirect representations of $L$ are trivial, i.e., if $L\cong L_1\leq \prod_{i\in I}L_i$ then necessarily $L\cong L_i$ for some index $i\in I$. In other words, a lattice is subdirectly irreducible if it is not subdirectly representable by ``simpler" lattices. Note that subdirectly irreducible lattices play a similar role with respect to subdirect products of lattices as primes with respect to multiplication of integers. 

Recall the well-known fact that every lattice $L$ is isomorphic to a subdirect product of subdirectly irreducible lattices, which are homomorphic images of $L$. From this point of view, the subdirectly irreducible lattices can be considered to have a simpler structure than the former lattice. 

As subdirectly irreducible factors of a subdirect representation can be found within the homomorphic images, it follows that every finite lattice is isomorphic to a subdirect product of finite numbers of finite subdirectly irreducible lattices.

For the sake of simplicity, when dealing with a subdirect representation, we limit our attention to the finite case.

First, observe the following two simple but important facts. If $f_1\colon L\to M$, $g_1\colon M\to L$ and $f_2\colon M\to K$, $g_2\colon K\to M$ are monotone Galois connections, then the composition $f_1\circ f_2\colon L\to K$, $g_2\circ g_1\colon K\to L$ is a monotone Galois connection as well. Indeed, for all $x\in L$ and $y\in K$ we have
$$ x\leq g_1(g_2(y))\quad\mbox{iff}\quad f_1(x)\leq g_2(y)\quad\mbox{iff}\quad f_2(f_1(x))\leq y.$$

Further, given a finite lattice $L$ and a sublattice $M\subseteq L$, let $c_M$ and $i_M$ be the corresponding closure and interior operators, respectively. Note that $M$ is both closure and interior system on $L$, since $M$ is a sublattice. Then the pair $(c_M,\mathrm{id}_M)$, considering $c_M$ as a mapping $c_M\colon L\to M$ and $\mathrm{id}_M\colon M\to L$ being the identity inclusion of $M$ into $L$, forms a monotone Galois connection between $L$ and its sublattice $M$. This follows from the defining properties of the closure operators, in fact for any $x\in L$ and  $y\in M$ we obtain
$$ x\leq \mathrm{id}_M(y)=y \quad\mbox{iff}\quad c_M(x) \leq c_M(y)=y. $$
Similarly, the pair $(\mathrm{id}_M,i_M)$, $i_M\colon L\to M$, forms a monotone Galois connection between $M$ and $L$, since for all $x\in M$ and $y\in L$
$$ x\leq i_M(y) \quad\mbox{iff}\quad \mathrm{id}_M(x)=x\leq i_M(y)\leq y.$$

\begin{lemma}
Let $L$ be a finite lattice and $M\subseteq L$ be a sublattice. Then any $\bigvee$-preserving function $f\colon M\to M$ is given by
\begin{equation}\label{eq31}
f(x)=c_M(F(x)), \ \mbox{for all}\ x\in M,
\end{equation}
where $F\colon L\to L$ is a $\bigvee$-preserving mapping. 

In this case, the upper adjoint $g\colon M\to M$ of $f$ is given by 
\begin{equation}\label{eq32}
g(y)=i_M(G(y)), \ \mbox{for all}\ y\in M,
\end{equation}
where $G\colon L\to L$ is the upper adjoint of $F$.
\end{lemma}

\begin{proof}
Obviously, if $f$ and $g$ are given by \eqref{eq31} and \eqref{eq32} respectively, then 
$$f=\mathrm{id}_M\circ F\circ c_M \quad\mbox{and}\quad g= \mathrm{id}_M\circ G \circ i_M. $$
Since the pairs of mappings $(\mathrm{id}_M,i_M)$, $(F,G)$ and $(c_M,\mathrm{id}_M)$ form monotone Galois connections, their composition pair $(f,g)$ forms a monotone Galois connection as well.

Further, assume that $f\colon M\to M$ is a $\bigvee$-preserving mapping. Then we can extend $f$ to a $\bigvee$-preserving mapping $F$ with the domain $L$ as follows:
$$F(x)=f(c_M(x)), \ \mbox{for all}\ x\in L.$$
The mapping $F\colon L\to M\subseteq L$ is a composition of $\bigvee$-preserving mappings since $c_M$ is the lower adjoint of $\mathrm{id}_M$, and thus it is $\bigvee$-preserving. Moreover, for all $x\in M$ we have $c_M(x)=x$ and $f(x)\in M$, which yields 
$$ c_M(F(x))=c_M(f(c_M(x)))=c_M(f(x))=f(x).$$ 
Finally, this shows that $f$ can be expressed by \eqref{eq31}. 
\end{proof}

\begin{theorem}
Let $L$ be a finite lattice, $\{L_i\mid i\in I\}$ be a finite family of finite lattices such that $L\leq \prod_{i\in I}L_i$ and let $c_L$, $i_L$ be the corresponding closure and interior operators on $\prod_{i\in I}L_i$. Then the following conditions are equivalent:
\begin{enumerate}
\item $f\colon L^n\to L$ is a $\bigvee$-preserving aggregation function. 
\item For each $l\in\{1,\dots,n\}$ there is a system $\{f_{ij}^l\colon L_i\to L_j\mid i,j\in I\}$ of $\bigvee$-preserving mappings satisfying $\bigvee_{l=1}^n\bigvee_{i\in I}f_{ij}^l(1)=1$. Moreover, $f(\mathbf{x})=\bigvee_{l=1}^n c_L(F_l(\mathbf{x}(i)))$ for all $\mathbf{x}\in L^n$, and for each $l\in\{1,\dots,n\}$, $F_l\colon \prod_{i\in I}L_i\to \prod_{i\in I}L_i$ are given by \eqref{eq1}.
\end{enumerate}
\end{theorem}

\begin{example}

Consider the lattice $L$ from Example \ref{ex1}. It is the well-known fact that the only subdirectly irreducible distributive lattice is isomorphic to the two element chain $\mathbf{2}$. As $L$ is distributive, it has a subdirect representation into the direct power of two element chains. Fig. \ref{fig3} shows such a subdirect representation of $L$ in $\mathbf{2}^3$.  

\begin{figure}
\begin{center}
\includegraphics[scale=0.75]{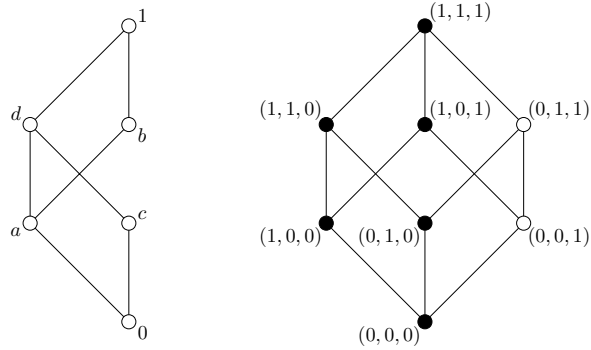}
\caption{A subdirect representation of the lattice $L$ in $\mathbf{2}^3$.}
\label{fig3}
\end{center}
\end{figure}

With respect to this representation, any $\bigvee$-preserving mapping is determined by a system of closure-interior systems on $\mathbf{2}$. We present such a system, corresponding to the $\bigvee$-preserving mapping determined by the pair $(S_1,T_1)$ from Example \ref{ex1}. Particularly, as $L$ is subdirectly represented in $\mathbf{2}^3$, we have $3\cdot 3=9$ pairs of closure-interior systems. They are depicted in Fig. \ref{fig4}. The system in the $i$-th row and in the $j$-th column corresponds to the mapping between the $i$-th factor and the $j$-th factor within the subdirect representation of $L$. 

\begin{figure}
\begin{center}
\includegraphics[scale=1]{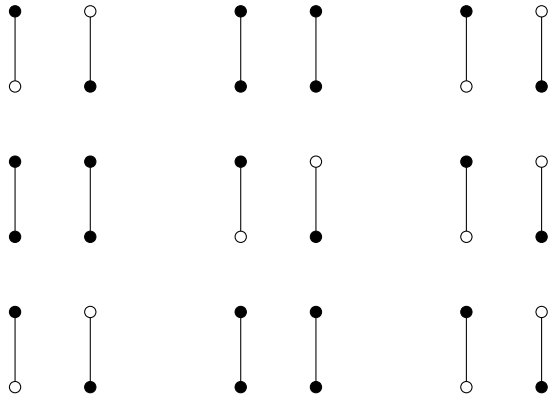}
\caption{``Subdirect decomposition" of the mapping corresponding to $(S_1,T_1)$ pair of Example \ref{ex1}.}
\label{fig4}
\end{center}
\end{figure}

\end{example}

Let us note, that in a similar way the subdirect decomposition of the mapping corresponding to $(S_2,T_2)$ pair can be obtained. Hence the whole aggregation function $f$ from Example \ref{ex1} can be characterized as a $6\times 3$ matrix consisting of isomorphic closure-interior pairs on $\mathbf{2}$.

\section{Relationship to the FCA based biclustering methods} 

In this section we briefly discuss possible applications of the results from the previous section within certain biclustering fuzzy FCA-based methods.

As it is common in cluster analysis, an $n\times m$ data matrix $R$ is given by objects $B$, attributes $A$ and entries $R(b,a)$. The primary aim of bicluster analysis is to identify subgroups of objects $X\subseteq B$ which are as similar as possible to each other with respect to some subset of attributes $Y\subseteq A$, and different as much as possible to the rest of objects and attributes. Bicluster is then formally defined as a pair $(X,Y)$. 

Such a relatively wide definition of biclustering certainly fulfills the classical FCA, cf. \cite{Ifca} where some link between FCA and biclustering can be found. In this case a data matrix $R$ contains only $0$ - $1$ values and it can be formally seen as an incidence relation between the objects and attributes, i.e., $R\subseteq B\times A$. Then $(b,a)\in R$ or equivalently $R(b,a)=1$ is interpreted as ``an object $b$ has an attribute $a$".  
The biclusters, so-called formal concepts, correspond to the maximal rectangles in the data sets. Given a formal concept $(X,Y)$, $Y$ is a subset of all attributes shared by all objects of $X$, while $X$ is a subset of objects sharing all attributes of $Y$. This ``sharing" of attributes represents the essence of the similarity, mentioned in the definition of biclustering. 
It turns out that formal concepts can be defined via concept forming operators, acting between the power set $\mathbf{P}(B)$ of objects and the power set $\mathbf{P}(A)$ of attributes. These operators are induced by the relation $R$ 
$$ X^{\prime}=\{a\in A\mid (b,a)\in R \ \mbox{for all}\ b\in X\}$$
$$ Y^{\backprime}=\{b\in B\mid (b,a)\in R\ \mbox{for all}\ a\in Y\}.$$
The formal concepts are precisely the fixed points of the operators, i.e., pairs $(X,Y)$ fulfilling $X^{\prime}=Y$ and $X=Y^{\backprime}$. Such defined concept forming operators form an antitone Galois connection between $\mathbf{P}(B)$ and $\mathbf{P}(A)$, cf. \cite{GW}.

Having this in mind, several biclustering-like fuzzy approaches to FCA were proposed, either with the help of antitone or monotone Galois connections \cite{B01,Kr1,MedN3,MedN2,MedN1}.

To mention an example of such fuzzy concept forming operators, consider the so-called monotone $\mathbf{L}$-Galois connections, introduced in \cite{GePo}. In this case, $\mathbf{L}=(L,\wedge,\vee,\otimes,\rightarrow,0,1)$ is a complete commutative residuated lattice and $R\colon B\times A\to L$ is an $\mathbf{L}$-relation. The concept forming operators $\uparrow\colon L^B\to L^A$ and $\downarrow \colon L^A\to L^B$ are given by
$${\uparrow}(\mathbf{x})(a)=\bigvee_{b\in B}\mathbf{x}(b)\otimes R(b,a), $$
$${\downarrow}(\mathbf{y})(b)=\bigwedge_{a\in A} R(b,a)\rightarrow \mathbf{y}(a).$$

These concept forming operators are typically involved when some type of fuzzy logic, with $\mathbf{L}$ as the truth value structure, is used for evaluation membership degrees of particular attributes. The operation $\otimes$ is a fuzzy counterpart of the classical logical conjunction, while $\rightarrow$ stands for a fuzzy implication. From a fuzzy logical point of view, the concept forming operators admit the following interpretation: ${\uparrow}(\mathbf{x})(a)$ represents the truth degree of the proposition ``there exists an object $b\in X$ having an attribute $a$", where the fuzzy subset $\mathbf{x}$ corresponds to $X$, and ${\downarrow}(\mathbf{y})(b)$ represents the truth degree of the proposition ``for all attributes $a\in A$, $a\in Y$ provided an object $b$ is in relation $R$ with $a$", $\mathbf{y}$ taking the role of $Y$.

In each residuated lattice, the two fuzzy connectives $\otimes$ and $\rightarrow$ are related by adjoint property
$$ x\otimes a\leq y \quad\mbox{iff}\quad x\leq a\rightarrow y.$$
Hence each $a\in L$ determines a $\bigvee$-preserving mapping $f_a\colon L\to L$, $f_a(x)=x\otimes a$ for all $x\in L$, with the upper adjoint $g_a(y)=a\rightarrow y$ for all $y\in L$.
Observe that concept forming operators $\uparrow$ and $\downarrow$ are defined by \eqref{eq1} and \eqref{eq2}, respectively.

The previous considerations allow to define a slightly modified, but more general approach, not connected with any fuzzy logic framework. Let $L$ be a fixed complete lattice, $B$ be a set of objects and $A$ be a set of attributes characterizing particular objects. As a basic input, consider a data table in the form of a many-valued binary relation $R\colon B\times A\to V$, where $V$ represents some set of possible alternatives for a characterization of objects from $B$ by particular attributes from $A$. To induce concept forming operators from such input data, we use a mapping from the set $V$ into the set of $\bigvee$-preserving mappings, where $f_a\colon L\to L$ denotes a mapping associated to an element $a\in V$. Then applying Theorem \ref{thm11} and its corollary we obtain the mappings $F\colon L^B\to L^A$ and $G\colon L^A\to L^B$ defined by
\begin{equation}\label{eq41}
F(\mathbf{x})(a)=\bigvee_{b\in X}f_{R(b,a)}(\mathbf{x}(b)),
\end{equation}
for all $\mathbf{x}\in L^B$, and
\begin{equation}\label{eq42}
G(\mathbf{y})(b)=\bigwedge_{a\in A}g_{R(b,a)}(\mathbf{y}(a)),
\end{equation}
for all $\mathbf{y}\in L^A$, form a monotone Galois connection between them. The biclusters, or fuzzy formal concepts, are defined as the fixed points of these operators, i.e. $(\mathbf{x},\mathbf{y})$ is a fuzzy formal concept if $F(\mathbf{x})=\mathbf{y}$ and $\mathbf{x}=G(\mathbf{y})$.

From Proposition \ref{prop3} we obtain that the set of all concepts, partially ordered by $(\mathbf{x}_1,\mathbf{y}_1)\leq (\mathbf{x}_2,\mathbf{y}_2)$ if $\mathbf{x}_1\leq\mathbf{x}_2$ (or equivalently, $\mathbf{y}_1\leq\mathbf{y}_2$), has a lattice structure. Particularly, this concept lattice is isomorphic to the induced interior system $\mathrm{Rng}(F)$ on $L^A$, which is also isomorphic to the induced closure system $\mathrm{Rng}(G)$ on $L^B$.

Taking into account a natural condition that the top elements of $L^B$ and $L^A$ should form a formal concept, we obtain that for each $a\in A$ the composition $F\circ \pi_a\colon L^B\to L$ of $F$ and the projection map $\pi_a$ forms a $\bigvee$-preserving aggregation function.
Hence the basic concept forming operator can also be seen as a system of $\left|A\right|$ $\bigvee$-preserving mappings, where any of them is in some sense determined by the values $R(b,a)$, $b$ varying through the set of all objects $B$. Given a fixed attribute $a\in A$, the mapping $F\circ \pi_a$ assigning to each $\mathbf{x}\in L^B$ the value $F(\mathbf{x})(a)=\bigvee_{b\in B} f_{R(b,a)}(\mathbf{x}(b))$ can be seen as some kind of weighted supremum of the values $\{\mathbf{x}(b)\mid b\in B\}$. Consequently, formal concepts can be studied and interpreted within many other theories, where these types of functions play an important role, e.g., in multicriterial decision support. 
Also the mapping $F\circ \pi_a$ can be understood as some $L$-valued possibility measure on $L$-fuzzy sets, being a fuzzy analogy of possibility measures. Such view allows to consider about FCA-based clustering methods in the realm of the possibility theory. These different perspectives on the mentioned fuzzy FCA-based biclustering method can be useful in order to incorporate other types of information usually available for the considered data. 

\section{Conclusion}

Formal concepts can be studied and interpreted within many different
theories, where these types of clusters play an important role, e.g., in
multicriteria decision support, or in possibility theory. In this
contribution, we have focused on links between aggregation functions acting on
complete lattices and formal concept analysis. In particular, we have
elaborated a description of sup- (inf-) preserving aggregation functions, thus
generalizing several particular results known from the literature (such as the
structure of sup- and inf-preserving triangular norms and conorms on product
lattices characterized in \cite{BM}). We have also discussed possible applications of
our results  within certain biclustering fuzzy FCA-based methods. We believe
that our approach and examples of applications of the general methods of
formal concept analysis will expand to several new areas, offering them a
powerful tool.

\section*{Acknowledgment}

The first author was supported by the international project Austrian Science Fund (FWF)-Grant Agency of the Czech Republic (GA\v{C}R) no. 15-34697L; the second author by the Slovak Research and Development Agency under contract APVV-14-0013; the third author by the IGA project of the faculty of Science Palack\'y University Olomouc no. PrF2016006 and by the Slovak VEGA Grant no. 2/0044/16.

\end{document}